\documentclass[11pt]{article}
\usepackage{a4wide}
\usepackage{authblk}
\usepackage{amsmath,amssymb,amsfonts,latexsym,stmaryrd}
\usepackage{mathtools}
\usepackage[mathscr]{eucal}
\usepackage[PostScript=dvips]{diagrams}
\usepackage{stmaryrd}
\usepackage{amsthm}
\usepackage{tikz-cd}

\newtheorem{theorem}{Theorem}
\newtheorem{proposition}[theorem]{Proposition}

\newcommand{\Ck}{\mathbb{C}_k}

\newcommand{\id}{\mathsf{id}}

\newcommand{\ie}{\textit{i.e.}~}
\newcommand{\Fraisse}{Fra\"{i}ss\'{e}~}
\newcommand{\Ek}{\mathbb{E}_{k}}
\newcommand{\Pk}{\mathbb{P}_{k}}
\newcommand{\vsa}{\vspace{.1in}}

\newcommand{\Struct}{\textsf{Structure}}
\newcommand{\Power}{\textsf{Power}}

\newcommand{\As}{\mathscr{A}}
\newcommand{\Bs}{\mathscr{B}}

\newcommand{\sg}{\sigma}
\newcommand{\RA}{R^{\As}}
\newcommand{\RB}{R^{\Bs}}

\newcommand{\IMP}{\; \Rightarrow \;}

\newcommand{\CS}{\mathcal{R}(\sg)}

\newcommand{\preford}{\sqsubseteq}

\newcommand{\IFF}{\Longleftrightarrow}
\newcommand{\rarr}{\rightarrow}
\newcommand{\Count}{\#}
\newcommand{\LL}{\mathcal{L}}

\newcommand{\Lk}{\mathcal{L}_k}

\newcommand{\Lck}{\mathcal{L}_{k}(\Count)}
\newcommand{\ELk}{\exists\mathcal{L}_{k}}

\newcommand{\eqLk}{\equiv^{\Lk}}
\newcommand{\eqELk}{\equiv^{\ELk}}
\newcommand{\eqLck}{\equiv^{\Lck}}

\newcommand{\KK}{\mathcal{K}}
\newcommand{\eqL}{\equiv^{\LL}}
\newcommand{\vphi}{\varphi}
\newcommand{\iffdef}{\;\; \stackrel{\Delta}{\IFF} \;\;}
\newcommand{\GG}{\mathsf{G}}

\newcommand{\Kl}{\mathsf{Kl}}
\newcommand{\epsA}{\varepsilon_{\As}}
\newcommand{\eqak}{\rightleftarrows_{k}}
\newcommand{\eqbk}{\leftrightarrow_{k}}
\newcommand{\eqck}{\cong_{\Kl(\Ck)}}
\newcommand{\eqaCk}{\rightleftarrows_{k}^{\mathbb{C}}}
\newcommand{\eqbCk}{\leftrightarrow_{k}^{\mathbb{C}}}
\newcommand{\eqcCk}{\cong_{k}^{\mathbb{C}}}
\newcommand{\Alk}{A^{\leq k}}
\newcommand{\Blk}{B^{\leq k}}

\newcommand{\REk}{R^{\Ek \As}}
\newcommand{\Mk}{\mathbb{M}_{k}}
\newcommand{\Ralph}{R_{\alpha}}

\newcommand{\RMA}{R^{\Mk (\As, a)}}

\newcommand{\kset}{\mathbf{k}}

\newcommand{\elen}{\exists_{\leq n}}
\newcommand{\egen}{\exists_{\geq n}}

\newcommand{\pow}{\mathscr{P}}
\newcommand{\comp}{{\uparrow}}
\newcommand{\adj}{\frown}

\newcommand{\hgt}{\mathsf{ht}}
\newcommand{\td}{\mathsf{td}}
\newcommand{\tw}{\mathsf{tw}}

\newcommand{\Gf}{\mathcal{G}}
\newcommand{\cnE}{\kappa^{\mathbb{E}}}
\newcommand{\cnP}{\kappa^{\mathbb{P}}}

\newcommand{\lbfn}{\lambda}
\newcommand{\pth}{\mathsf{path}}

\newcommand{\es}{\varnothing}
\newcommand{\WAB}{\mathsf{W}_{\As,\Bs}}
\newcommand{\SAB}{\mathcal{S}(\As,\Bs)}
\newcommand{\SBA}{\mathcal{S}(\Bs,\As)}

\newcommand{\rcvr}{\succ}

\newcommand{\arX}{\rightarrow_{X}}
\newcommand{\arnX}{\rightarrow_{X}^{n}}

\newcommand{\barX}{\rightleftarrows_{X}}
\newcommand{\barnX}{\rightleftarrows_{X}^{n}}

\newcommand{\barnK}{\rightleftarrows_{K}^{n}}
\newcommand{\barn}{\rightleftarrows^{n}}
\newcommand{\rln}{\rightleftarrows^n}
\newcommand{\rla}{\rightleftarrows}
\newcommand{\emp}{\varnothing}
\newcommand{\emb}{\rightarrowtail}

\newcommand{\poK}{+_{K}}

\newcommand{\Approxn}{\mathsf{Approx}_{n}}
\newcommand{\EC}{\Theta}
\newcommand{\iKCh}{i^{K}_{C, h}}

\newcommand{\CSn}{\mathcal{R}_{\sg}^{n}}

\begin{document}

\title{Whither Semantics?}
\author{Samson Abramsky\thanks{samson.abramsky@cs.ox.ac.uk}~}
\affil{Department of Computer Science, University of Oxford}
\date{}
\maketitle

\begin{abstract}
We discuss how mathematical semantics has evolved, and suggest some new directions for future work.
As an example, we discuss some recent work on encapsulating model comparison games as comonads, in the context of finite model theory.
\end{abstract}


\section{Introduction}

Maurice Nivat was one of the founding fathers of  Theoretical Computer Science in Europe, both through his scientific achievements, and his community-building work, including establishing the Journal of Theoretical Computer Science, and as one of the founders of the European Association for Theoretical Computer Science, and the ICALP conference. The very name of ICALP -- International Colloquium on Automata, Languages and Programming -- shows the two main scientific lines to which Nivat himself made pioneering contributions: formal languages and automata on the one hand, and semantics of programming languages on the other.  Having started his scientific career in the former, Nivat had the breadth of interest and perspective to see the importance and depth of the work on mathematical semantics of programming languages of Dana Scott, Robin Milner, Rod Burstall, \textit{et al.} 
He founded a French school of algebraic semantics of recursive program schemes, and did much to create the conditions and possibilities which allowed for the development by others of work in $\lambda$-calculus, type theory, and programming language semantics.                                

In my own entry to the programming languages world, I encountered Nivat's work, including \cite{nivat1979infinite}, which gives an elegant introduction to how infinite objects arise in the semantics of computation, and the mathematical theory around them. Another paper \cite{nivat1980non}, which discusses the semantics of non-deterministic program schemes, highlighting some key issues with striking examples, directly influenced my  early work in \cite{abramsky1983semantic}.

On a personal level, at an early stage in my career, Maurice invited me to be an editor of TCS. I believe that he similarly encouraged and supported many young researchers.

In this short essay, I would like to look at some broad features of how the field of semantics has developed, and to ask some questions about where it is, or should be, going. This will lead to a brief overview of some current work. 

\section{The evolving project of semantics}

The purposes of programming language semantics, as originally conceived, can be summarized as follows:
\begin{itemize}
\item To give a mathematically precise, implementation independent specification of the meanings of programs.
\item This can then serve as a basis for, e.g.
\begin{itemize}
\item soundness of program logics
\item soundness of type systems
\item proving program equivalences
\item compiler correctness.
\end{itemize}

\end{itemize}

The first ``formal semantics'' were given in terms of abstract machines \cite{wegner1972vienna,landin1965abstract}.
Then denotational semantics, based on domain theory, appeared -- \emph{mathematical semantics} (Scott and Strachey) \cite{scott1970outline,scott1971toward}.
The contrast between this ``mathematical semantics'' and abstract machine-based semantics was rather clear.

Then ``Structural Operational Semantics'' (Plotkin) \cite{plotkin1981structural}, ``Natural Semantics'' (Kahn) \cite{kahn1987natural}, and their variants appeared on the scene.

A structural operational  semantics is given by a syntax-directed inductive definition of program behaviour, expressed in terms of \emph{directly operationally meaningful} data, and \emph{without} invoking abstract mathematical spaces (e.g. function spaces).
It is formally rigorous, and can serve as a mathematically precise specification of a programming language.

It should be said that this form of semantics is now probably the ``industry standard''  in the community doing formal verification of programming language properties and directly addressing the objectives we listed above, as represented in conferences such as ACM POPL.

Much of this activity, and the community doing it, is rather disjoint from the ongoing work in mathematical semantics.
Is it then time to reconsider the purposes of mathematical semantics, and hence where it should be going?

\subsection{Mathematical Semantics}

The term ``mathematical semantics'', which may seem old-fashioned, is preferable, for the purposes of our discussion, to ``denotational semantics'', which is perhaps too tied to set-theoretic forms of semantics. Various forms of categorical semantics offer \emph{structures} rather than denotations as their primary features.

Our question can be posed as follows: given that many of the original purposes of semantics, as listed above, are very adequately fulfilled by various forms of structural operational semantics, what is the compelling motivation for mathematical semantics?
This is a rather fundamental question, yet an answer is hard to give in precise terms! 

It may be useful to illustrate this general question with a classic example: the untyped $\lambda$-calculus \cite{barendregt1984lambda}.
It was known since the 1930's, by the Church-Rosser theorem, that the calculus is consistent.
We can express this by the construction of a non-trivial term model. 
On the other hand, the introduction of domain-theoretic models of the $\lambda$-calculus by Scott in 1969 revolutionized the subject. 
We can ask: what did these models add?

In our view, the most important feature of domain theory is that it provides a \emph{semantic paradigm}, which exists independently of any particular language or calculus. 
This paradigm has proved extremely fruitful, leading to many new questions and developments.

\subsection{A vision for mathematical semantics?}

More generally, we can say that mathematical semantics can lead to new semantic paradigms.
These can lead to new languages, logics, tools, new kinds of questions, new structural theories.
For an early example in the field of untyped $\lambda$-calculus, one can refer to the classic treatise \cite{barendregt1984lambda}. For recent examples, see e.g.~\cite{vakar2019domain,tasson2018models}.

This in turn, suggests a  \emph{semantics dilemma}: 
\begin{center}
\fbox{To follow or to lead?}
\end{center}
That is, should mathematical semantics still be conceived as following in the track of pre-existing languages, trying to explain their novel features, and to provide firm foundations for them? Or should it be seen as operating in a more autonomous fashion, developing new semantic paradigms, which may then give rise to new languages?

Indeed, a number of discussions by pioneers of the field relate to this question:
\begin{itemize}
\item Robin Milner gave an invited lecture at MFPS in 2000, in which he addressed exactly this point, and proposed the term ``Information Dynamics''
as a possible substitute for ``semantics'', to emphasize this autonomous character of the study of mathematical structures in computation.
Unfortunately, we have not been able to find any published record of his lecture.

\item The following quotation from Tony Hoare \cite{hoare1980model} beautifully expresses the point:
\begin{quote}
\textit{The primary objective of this paper is to give a simple
mathematical model for communicating sequential processes.}

\textit{As the exposition unfolds, the examples begin to look like programs,
and the notations begin to look like a programming language.
Thus the design of a language seems to emerge naturally from its
formal definition, in an intellectually pleasing fashion.}
\end{quote}

\item Finally, a tantalizing footnote from a fascinating (and too little known) paper by Peter Landin \cite{landin1969program} (see \cite{DBLP:journals/corr/Abramsky14d} for a discussion relating this paper to later developments):
\begin{quote}
\textit{For some years I have aspired to syntax-free programming \ldots}
\end{quote}
\end{itemize}

\subsection{Semantic Dreams}

In this brief essay, we would like to suggest a still bolder and more expansive vision for mathematical semantics: a wider application of its paradigms and tools, both within Computer Science and beyond:
\begin{itemize}
\item Beyond Computer Science, there are many opportunities to use the compositional methods of semantics to great effect: we refer for examples to ongoing work in categorical quantum mechanics \cite{abramsky2004categorical,abramsky2008categorical}, the sheaf-theoretic approach to contextuality and its applications \cite{abramsky2011sheaf,abramsky2015contextuality}, and work in linguistics \cite{abramsky2014semantic}, game theory \cite{abramsky2017coalgebraic,ghani2018compositional}, networks and signal-flow graphs \cite{baez2017network,bonchi2014categorical}, etc.

\item Within Computer Science itself, we see this expansive view of semantics as potentially contributing to a unification of our field, across a major current divide. This will be the topic of the remainder of this article.
\end{itemize}

\section{Structure \textit{vs} Power: The Great Divide}
The logical foundations of computer science have undergone immense development over the past seven decades. However, the current state of the art shows a remarkable \emph{great divide}, into two  streams:
\begin{itemize}
\item the connections between logic and algorithms and complexity, studying topics such as logic and automata, finite model theory, descriptive complexity, database theory and constraint satisfaction
\item the study of semantics and type theory, using tools from category theory, $\lambda$-calculus and proof theory.
\end{itemize}
These form distinct bodies of work, largely pursued by disjoint communities, with different methods, concepts and technical languages, and with little communication, or even mutual comprehension, between them.
Some pieces of work bring these streams together, but they are rare.

Each of these streams reflects and informs important areas of more applied work. Semantics and type theory have strongly influenced the development of modern programming languages and their associated tools, and also of interactive theorem provers and proof assistants; while work in databases, computer-assisted verification, constraint satisfaction etc.~shows the influence of the study of logic in relation to algorithms and complexity.

This dichotomy reflects two very different views of what the fundamental features of computation are: one focussing on \emph{structure} and \emph{compositionality}, the other on \emph{expressiveness} and \emph{efficiency}. For brevity, we can call these \Struct~and \Power.

\begin{description}
\item[Structure] Compositionality and semantics, addressing the question of mastering the complexity of computer systems and software.

\item[Power] Expressiveness and complexity,
addressing the question of how we can harness the power of computation and recognize its limits.
\end{description}

A \emph{shocking fact}: the current state of the art is almost \emph{disjoint communities} of researchers studying \Struct~and \Power~respectively, with no common technical language or tools.
This is a major obstacle to fundamental progress in Computer Science.

We can make an analogy (emphasizing that it is \emph{only} an analogy) with the Grothendieck program in algebraic geometry. The (very abstract) tools developed there were ultimately critical for concrete results, e.g.~the Wiles proof of the Fermat theorem.

It is instructive to quote one renowned number theorist, quoting another --
Mazur quoting Lenstra \cite{mazurlen1997}:
\begin{quote}
\textit{twenty years ago he was firm in his conviction that he DID want to solve Diophantine equations, and that he DID NOT wish to represent functors -- and now he is amused to discover himself representing functors in order to solve Diophantine equations!}
\end{quote}

The dream is to use \emph{structural methods} to solve \emph{hard problems}.
Can it be done? Only if we are bold enough to try!

\subsection*{A case study}

As a case study for this theme, we shall give a brief overview of some recent work on relating categorical semantics, which exemplifies \Struct, to finite model theory, which exemplifies \Power.  This is based on the papers \cite{abramsky2017pebbling,DBLP:conf/csl/AbramskyS18}, and ongoing work with Nihil Shah and Tom Paine.
Readers wishing to see a more detailed account are referred to \cite{abramsky2017pebbling,DBLP:conf/csl/AbramskyS18}.
In particular, the presentation in \cite{abramsky2017pebbling}, while detailed, should be accessible to readers with minimal background in category theory.

\section{The setting}

Relational structures and the homomorphisms between them play a fundamental r\^{o}le in finite model theory, constraint satisfaction and database theory. The existence of a homomorphism $A \rarr B$ is an equivalent formulation of constraint satisfaction, and also equivalent to the preservation of existential positive sentences \cite{chandra1977optimal}. This setting also generalizes what has become a central perspective in graph theory \cite{hell2004graphs}.

A relational vocabulary $\sg$  is a family of relation symbols $R$, each of some arity $n > 0$.
A relational structure for $\sg$ is $\As = (A, \{ \RA \mid R \in \sg\}))$, where $\RA \subseteq A^n$ for $R \in \sg$ of arity $n$.
A homomorphism of $\sg$-structures $f : \As \rarr \Bs$ is a function $f : A \rarr B$ such that, for each relation $R \in \sg$ of arity $n$ and $(a_1, \ldots , a_{n}) \in A^{n}$:
\[  (a_1, \ldots , a_{n}) \in \RA \IMP (f(a_1), \ldots , f(a_n))) \in \RB . \]
Our setting will be $\CS$, the category of relational structures and homomorphisms.

\subsection*{Model theory and deception}
In a sense, the purpose of model theory is ``deception''.  It allows us to see structures not ``as they really are'', \ie up to isomorphism, but only up to \emph{definable properties}, where definability is relative to a logical language $\LL$. The key notion is \emph{logical equivalence} $\eqL$. Given structures $\As$, $\Bs$ over the same vocabulary:
\[ \As \eqL \Bs \iffdef \forall \vphi \in \LL. \; \As \models \vphi \; \IFF \; \Bs \models \vphi . \]
If a class of structures $\KK$ is definable in $\LL$, then it must be saturated under $\eqL$. Moreover, for a wide class of cases of interest in finite model theory, the converse holds \cite{kolaitis1992infinitary}.

The idea of syntax-independent characterizations of logical equivalence is quite a classical one in model theory, exemplified by the Keisler-Shelah theorem \cite{shelah1971every}.
It acquires additional significance in finite model theory, where model comparison games such as Ehrenfeucht-\Fraisse (EF)-games, pebble games and bisimulation games play a central role \cite{Libkin2004}.

The EF-game between $\As$ and $\Bs$ is played as follows. In the $i$'th round, Spoiler moves by choosing an element in $A$ or $B$; Duplicator responds by choosing an element in the other structure. Duplicator wins after $k$ rounds if the relation $\{ (a_i, b_i) \mid 1 \leq i \leq k \}$ is a partial isomorphism.

In the existential EF-game, Spoiler only plays in $\As$, and Duplicator responds in $\Bs$. The winning condition is that the relation is a partial homomorphism.

The Ehrenfeucht-\Fraisse theorem says that a winning strategy for Duplicator in the $k$-round EF game characterizes the equivalence $\eqLk$, where $\Lk$ is the fragment of first-order logic of formulas with quantifier rank $\leq k$.

Similarly, there are $k$-pebble games, and bismulation games played to depth $k$.

Pebble games are similar but subtly different to EF-games.
Spoiler moves by placing one from a fixed set of pebbles on an element of $\As$ or $\Bs$; Duplicator responds by placing their matching pebble on an element of the other structure.
Duplicator wins if after each round, the relation defined by the current positions of the pebbles is a partial isomorphism
Thus there is a ``sliding window'' on the structures, of fixed size. It is this size which bounds the resource, not the length of the play.

Whereas the $k$-round EF game corresponds to bounding the quantifier rank, $k$-pebble games correspond to bounding the number of variables which can be used in a formula.
Just as for EF-games, there is an existential-positive version, in which Spoiler only plays in $\As$, and Duplicator responds in $\Bs$.

Bisimulation games are localized to a current element of the universe (which is typically thought of as a ``world'' or a ``state'').
Spoiler can move from the current element in one structure to another element which is related to the current one by some binary relation (a ``transition''). Duplicator must respond with a matching move in the other structure.

\subsection*{A new perspective}
\begin{itemize}
\item We shall study these games, not as external artefacts, but as semantic constructions in their own right.
For each type of game $\GG$, and value of the resource parameter $k$, we shall define a corresponding construction~$\Ck$ on $\CS$. For each structure $\As$, we shall build a new structure $\Ck \As$. This new structure will represent the limited access to the underlying structure $\As$ which is available when playing the game with this level of resources.

\item The idea is that Duplicator strategies for the existential version of $\GG$-games from $\As$ to $\Bs$ will be recovered as homomorphisms $\Ck \As \to \Bs$.
Thus the notion of local approximation built into the  game is internalised into the category of $\sg$-structures and homomorphisms.

\item This leads to characterisations of a number of central concepts in
Finite Model Theory and combinatorics.
\end{itemize}

These constructions actually form \emph{comonads} on the category of relational structures. 
Monads and comonads are basic notions
of category theory which are widely used in semantics of
computation and in modern functional programming \cite{moggi1991notions,wadler1995monads,brookes1991computational,uustalu2008comonadic}.
We show that model-comparison games have a natural
comonadic formulation.

\subsection*{The EF comonad}

Given a structure $\As$, the universe of $\Ek \As$ is $\Alk$, the  non-empty sequences of length $\leq k$.
The counit map $\epsA : \Ek \As \to \As$ sends a sequence $[a_1, \ldots , a_n]$ to  $a_n$.

The key question is: how do we lift the relations on $\As$ to $\Ek \As$?

Given e.g.~a binary relation $R$, we define $\REk$ to the set of pairs $(s, t)$ such that
\begin{itemize}
\item $s \preford t$ or $t \preford s$ (in prefix order)
\item $\RA(\epsA(s), \epsA(t))$.
\end{itemize}
More generally, for each relation symbol $R$ of arity $n$, we define $\REk$ to be the set of $n$-tuples $(s_1, \ldots , s_n)$ of sequences which are pairwise comparable in the prefix ordering, and such that $\RA(\epsA s_1, \ldots , \epsA s_n)$.

\noindent
Given a homomorphism $f : \Ek \As \rarr \Bs$, we define the coextension $f^* : \Alk \rarr \Blk$ by 
\[ f^* [a_1, \ldots , a_j ] = [b_1, \ldots , b_j] , \]
where $b_i = f [a_1, \ldots , a_i]$, $1 \leq i \leq j$.

This is easily verified to yield a comonad on $\CS$ \cite{DBLP:conf/csl/AbramskyS18}.

Intuitively, an element of $\Alk$ represents a play in $\As$ of length $\leq k$. 
A coKleisli morphism $\Ek \As \rarr \Bs$ represents a Duplicator strategy for the existential Ehrenfeucht-\Fraisse game with $k$ rounds:
Spoiler plays only in $\As$, and $b_i = f [a_1, \ldots , a_i]$ represents Duplicator's response in $\Bs$ to the $i$'th move by Spoiler. 
 
The winning condition for Duplicator in this game is that, after $k$ rounds have been played, 
the induced relation $\{ (a_i, b_i) \mid 1 \leq i \leq k \}$ is a partial homomorphism from $\As$ to $\Bs$.

\begin{theorem}
\label{EFgamethm}
The following are equivalent:
\begin{enumerate}
\item There is a homomorphism $\Ek \As \rarr \Bs$.
\item Duplicator has a winning strategy for the existential Ehrenfeucht-\Fraisse game with $k$ rounds, played from $\As$ to $\Bs$.
\item For every existential positive sentence $\vphi$ with quantifier rank $\leq k$, $\As \models \vphi \IMP \Bs \models \vphi$.
\end{enumerate}
\end{theorem}

\subsection*{The pebbling comonad}

Given a structure $\As$, the universe of $\Pk \As$ is 
$(\kset \times A)^{+}$, the set of finite non-empty sequences of moves $(p, a)$, where $\kset \, := \, \{1, \ldots , k\}$. Note this will be infinite even if $\As$ is finite.
It is shown in \cite{abramsky2017pebbling} that this is essential.

The counit map $\epsA : \Pk \As \to \As$ sends a sequence $[(p_1,a_1), \ldots , (p_n,a_n)]$ to  $a_n$.

Again, the key question is, how do we lift the relations on $\As$ to $\Ek \As$?

Given e.g.~a binary relation $R$, we define $R^{\Pk \As}$ to the set of pairs $(s, t)$ such that
\begin{itemize}
\item $s \preford t$ or $t \preford s$ 
\item If $s \preford t$, then the pebble index of the last move in s does not appear in the
suffix of s in t; and symmetrically if $t \preford s$.
\item $\RA(\epsA(s), \epsA(t))$.
\end{itemize}

Given a homomorphism $f : \Pk \As \rarr \Bs$, we define the coextension $f^* : \Pk \As \rarr \Pk \Bs$ by 
\[ f^* [(p_1,a_1), \ldots , (p_j,a_j) ] = [(p_1,b_1), \ldots , (p_j,b_j)] , \]
where $b_i = f [(p_1,a_1), \ldots , (p_i,a_i)]$, $1 \leq i \leq j$.

Again, this is easily verified to yield a comonad on $\CS$, yielding entirely analogous results to those for $\Ek$.

\subsection*{The modal comonad}

The flexibility of the comonadic approach is illustrated by showing that it also covers the well-known construction of unfolding a Kripke structure into a tree (``unravelling'').

For the modal case, we assume that the relational vocabulary $\sg$ contains only symbols of arity at most 2. 
We can thus regard a $\sigma$-structure as a Kripke structure for a multi-modal logic.
If there are no unary symbols, such structures are exactly the labelled transition systems.

Modal logic localizes its notion of satisfaction in a structure to a world. 
We reflect this by using the category of \emph{pointed relational structures} $(\As, a)$. 
Objects are pairs $(\As, a)$ where $\As$ is a $\sg$-structure and $a \in A$. Morphisms $h : (\As, a) \rarr (\Bs, b)$ are $h : \As \rarr \Bs$ such that $h(a) = b$.

For $k>0$ we  define a comonad $\Mk$, where $\Mk (\As, a)$ corresponds to unravelling the structure $\As$, starting from $a$, to depth $k$.
The universe of $\Mk (\As, a)$ comprises $[a]$, which is the distinguished element,  together with all sequences of the form $[a_0, \alpha_1, a_1, \ldots , \alpha_j, a_{j}]$, where $a = a_0$, $1 \leq j  \leq k$, and $\RA_{\alpha_i}(a_i, a_{i+1})$, $0 \leq i < j$. 
For binary relations $\Ralph$, the interpretation is $\RMA_{\alpha}(s,t)$ iff for some $a' \in A$, $t = s[\alpha, a']$.

Verification that this is a comonad proceeds analogously to the previous cases.

\section{Logical equivalences}

For each of our three types of game, there are corresponding  fragments $\Lk$ of first-order logic:
\begin{itemize}
\item For Ehrenfeucht-\Fraisse games, $\Lk$ is the fragment of quantifier-rank $\leq k$.
\item For pebble games, $\Lk$ is the $k$-variable fragment.
\item For bismulation games over relational vocabularies with symbols of arity at most 2, $\Lk$ is the modal fragment  with modal depth $\leq k$.
\end{itemize}

In each case, we write 
\begin{itemize}
\item $\ELk$ for the existential positive fragment of $\Lk$

\item $\Lck$ for the extension of $\Lk$ with counting quantifiers $\elen$, $\egen$
\end{itemize}

We can generically define two equivalences based on our indexed comonads $\Ek$:
\begin{itemize}
\item $\As \eqaCk \Bs$ iff there are coKleisli morphisms $\Ck \As \rarr \Bs$ and $\Ck \Bs \rarr \As$. Note that there need be no relationship between these morphisms.
\vsa
\item $\As \eqcCk \Bs$ iff $\As$ and $\Bs$ are isomorphic in the coKleisli category $\Kl(\Ck)$. 
This means that there are morphisms $\Ck \As \rarr \Bs$ and $\Ck \Bs \rarr \As$ which are inverses of each other in $\Kl(\Ck)$.
\end{itemize}

To complete the picture, we need to show how to define a back-and-forth equivalence  $\eqbk$ which characterizes $\eqLk$ \emph{purely in terms of coKleisli morphisms}.
Our solution to this,  while not completely generic, is general enough to apply to all our game comonads -- so we subsume EF equivalence, bisimulation equivalence and pebble game equivalence as instances of a single construction.

An interesting feature is that it can be described in terms of approximations and fixpoints. We use total coKleisli morphisms to approximate partial isomorphisms ``from above''.

The definition is parameterized on a set $\WAB \, \subseteq \, \Ck A \times \Ck B$ of ``winning positions'' for each pair of structures $\As$, $\Bs$.
We assume, as is the case with our concrete comonadic constructions, that $\Ck A$ has a tree structure, writing $s' \rcvr s$ if $s$ is the unique immediate predecessor of $s'$.

We define the back-and-forth $\Ck$ game between $\As$ and $\Bs$ as follows:
\begin{itemize}
\item At the start of each round of the game, the position is specified by $(s, t) \in \Ck A \times \Ck B$. The initial position is $(\bot, \bot)$.  
\item Either Spoiler chooses some $s' \rcvr s$, and Duplicator responds with $t' \rcvr t$, resulting in  $(s', t')$; or Spoiler chooses  $t'' \rcvr t$ and Duplicator responds with $s'' \rcvr s$, resulting in $(s'',t'')$. 
\item Duplicator wins after $k$ rounds  if the resulting position $(s, t)$ is in $\WAB$.
\end{itemize}

This is essentially \emph{bisimulation up to $\WAB$}. 
By instantiating $\WAB$ appropriately, we obtain the equivalences corresponding to the EF, pebbling and bisimulation games.

For example, $\WAB^{\Ek}$ is the set of all $(s, t)$ which define a partial isomorphism.

\subsection*{Characterization by coKleisli morphisms}

We  define $\SAB$ to be the set of all functions $f : \Ck A \to B$ such that, for all $s \in \Ck A$, $(s, f^*(s)) \in \WAB$. 

A \emph{locally invertible pair} $(F, G)$ from $\As$ to $\Bs$ is a pair of sets $F \subseteq \SAB$, $G \subseteq \SBA$, satisfying the following conditions:
\begin{enumerate}
\item For all $f \in F$, $s \in \Ck A$, for some $g \in G$, $g^* f^*(s) = s$.
\item For all $g \in G$, $t \in \Ck B$, for some $f \in F$, $f^* g^*(t) = t$.
\end{enumerate}
Note that $F = \es$ iff $G = \es$.

We define $\As \eqbCk \Bs$ iff there is a non-empty locally invertible pair from $\As$ to $\Bs$.

\begin{proposition}
\label{pisoprop}
The following are equivalent:
\begin{enumerate}
\item $\As \eqbCk \Bs$.
\item There is a winning strategy for Duplicator in the $\Ck$ game between $\As$ and $\Bs$.
\end{enumerate}
\end{proposition}

\subsection*{A fixpoint characterization}

Write $S := \SAB$, $T := \SBA$.

Define set functions $\Gamma : \pow(S) \rarr \pow(T)$, $\Delta : \pow(T) \rarr \pow(S)$:  
\begin{align*}
\Gamma(F) & = \{ g \in T \mid \forall t \in \Ck B. \exists f \in F. \, f^* g^* t = t \}, \\
\Delta(G) & = \{ f \in S \mid \forall s \in \Ck A. \exists g \in G. \,  g^* f^* s = s \} .
\end{align*}

These functions are monotone. Moreover, a pair of sets $(F, G)$ is locally invertible iff $F \subseteq \Delta(G)$ and $G \subseteq \Gamma(F)$. 
These conditions in turn imply that $F \subseteq \Delta \Gamma(F)$, and if this holds, then we can set $G := \Gamma(F)$ to obtain a locally invertible pair $(F, G)$. 
Thus existence of a locally invertible pair is equivalent to the existence of non-empty $F$ such that $F \subseteq \Theta(F)$, where $\Theta = \Delta \Gamma$. 

Since $\Theta$ is monotone, by Knaster-Tarski this is equivalent to the greatest fixpoint of $\Theta$ being non-empty. (Note that $\Theta(\varnothing) = \varnothing$).

If $\As$ and $\Bs$ are finite, so is $S$, and we can construct the greatest fixpoint by a finite descending sequence
\[ S \supseteq \Theta(S) \supseteq \Theta^2(S) \supseteq \cdots \]
This fixpoint is non-empty iff $\As \eqbCk \Bs$.

\subsection*{Characterizations of logical equivalences}

\begin{theorem}[\cite{DBLP:conf/csl/AbramskyS18}]
For structures $\As$ and $\Bs$:
\begin{flushleft}
\begin{tabular}{llcl}
(1) & $\As \eqELk \Bs$ & $\; \IFF \;$ & $\As \eqak \Bs$. \\
(2) & $\As \eqLk \Bs$ & $\; \IFF \;$ & $\As \eqbk \Bs$. \\
(3) & $\As \eqLck \Bs$ & $\; \IFF \;$ & $\As \eqck \Bs$.
\end{tabular}
\end{flushleft}
\end{theorem}

Note that this is really a family of three theorems, one for each type of game. 
Thus in each case, we capture the salient logical equivalences in syntax-free, categorical form.

\section{Coalgebras and combinatorial parameters}

A beautiful feature of these comonads is that they let us capture crucial combinatorial parameters of structures using the indexed comonadic structure.

Conceptually, we can think of the morphisms $f : \Ck \As \to \Bs$ in the co-Kleisli category for $\Ck$ as those which only have to respect the $k$-local structure of $\As$.
The lower the value of $k$, the less information available to Spoiler, and the easier it is for Duplicator to have a winning strategy.
Equivalently, the easier it is to have a morphism from $\As$ to $\Bs$ in the co-Kleisli category.

What about morphisms $\As \to \Ck \Bs$?
Restricting the access to $\Bs$ makes it \emph{harder} for Duplicator to win the homomorphism game.

Another fundamental aspect of comonads is that they have an associated notion of \emph{coalgebra}. 
A coalgebra for a comonad $(G, \varepsilon, \delta)$ is a morphism $\alpha : A \to G A$ such that the following diagrams commute:
\[
\begin{diagram}
A & \rTo^{\alpha} & G A \\
\dTo^{\alpha} & & \dTo_{\delta_A} \\
G A & \rTo_{G \alpha} & G G A
\end{diagram}  \qquad \qquad
\begin{diagram}
A & \rTo^{\alpha} & G A \\
& \rdTo_{\id_A} & \dTo_{\epsA} \\
& & A
\end{diagram}
\]

We should only expect a coalgebra structure to exist when the $k$-local
information on A is sufficient to determine the structure of A.

Our use of indexed comonads $\Ck$ opens up a new kind of question for coalgebras. Given a structure $\As$, we can ask: what is the least value of $k$ such that a $\Ck$-coalgebra exists on $\As$?  
We call this the \emph{coalgebra number} of $\As$. 
We shall find that for each of our comonads, the coalgebra number is a significant combinatorial parameter of the structure.

\begin{theorem}
\begin{itemize}
\item For the pebbling comonad, the coalgebra number of $\As$ corresponds precisely to the \emph{tree-width} of $\As$.
\item For the Ehrenfeucht-\Fraisse comonad, the coalgebra number of $\As$ corresponds precisely to the \emph{tree-depth} of $\As$.
\item For the modal comonad, the coalgebra number of $(\As, a)$ corresponds precisely to the \emph{synchronization tree depth} of $a$ in $\As$. 
\end{itemize}
\end{theorem}

The main idea behind these results, as we shall now outline, is that coalgebras on $\As$ are in bijective correspondence with decompositions of $\As$ of the appropriate form. 
We thus obtain categorical characterizations of these key combinatorial parameters.

\subsection*{Tree depth and the Ehrenfeucht-\Fraisse comonad}

A graph is $G = (V, {\adj})$, where $V$ is the set of vertices, and $\adj$ is the adjacency relation, which is symmetric and irreflexive.
The comparability relation on a poset  $(P, {\leq})$ is $x \comp y$ iff  $x \leq y$ or $y \leq x$.  A chain in a poset $(P, {\leq})$ is  a subset $C \subseteq P$ such that, for all $x, y \in C$, $x \comp y$.
A \emph{forest} is a poset $(F, {\leq})$ such that, for all $x \in F$, the  set of predecessors  is a finite chain. 
The height $\hgt(F)$ of a forest $F$ is $\sup_{C} | C |$, where $C$ ranges over chains in $F$.
A forest cover for $G$ is a forest $(F, {\leq})$ such that $V \subseteq F$, and if $v \adj v'$, then $v \comp v'$.
The tree-depth $\td(G)$ \cite{nevsetvril2006tree} is defined to be $\min_{F} \hgt(F)$, where $F$ ranges over forest covers of $G$.

Given a $\sg$-structure $\As$, the Gaifman graph $\Gf(\As)$ is $(A, \adj)$, where $a \adj a'$ iff for some relation $R \in \sg$, for some $(a_1, \ldots , a_n) \in \RA$, $a = a_i$, $a' = a_j$, $i \neq j$. The tree-depth of $\As$ is $\td(\Gf(\As))$.

\begin{theorem}
\label{fcth}
Let $\As$ be a finite $\sg$-structure, and $k>0$. There is a bijective correspondence between
\begin{enumerate}
\item $\Ek$-coalgebras $\alpha : \As \rarr \Ek \As$.
\item Forest covers of $\Gf(\As)$ of height $\leq k$.
\end{enumerate}
\end{theorem}

\begin{proof} (Outline)

Suppose that $\alpha : \As \to \Ek \As$ is a coalgebra. For $a \in A$, let $\alpha(a) = [a_1, \ldots , a_j]$. 
\begin{itemize}
\item The first coalgebra equation says that $\alpha(a_i) = [a_1, \ldots , a_i]$, $1 \leq i \leq j$. 
\item The second says that $a_j = a$. 
\end{itemize}
Thus $\alpha$ is an injective map whose image is a prefix-closed subset of $\Alk$. 

Defining $a \leq a'$ iff $\alpha(a) \preford \alpha(a')$ yields a forest order on $A$, of height $\leq k$. 
If $a \adj a'$ in $\Gf(\As)$, the fact that $\alpha$ is a homomorphism implies $a \comp a'$. 

Thus $(A, {\leq})$ is a forest cover of $\As$, of height $\leq k$.
\end{proof}

As an easy consequence, we obtain:

\begin{theorem}
For all finite structures $\As$: $\td(\As) \, = \, \cnE(\As)$.
\end{theorem}

\subsection*{Tree width}

A tree $(T, {\leq})$ is a forest with a least element (the root).
The unique path from $x$ to $x'$ is  the set
$\pth(x, x') := [x \wedge x', x] \cup [x \wedge x', x']$, where we use interval notation: $[y, y'] := \{ z \in T \mid y \leq z \leq y' \}$. 

A tree decomposition of a graph $G = (V, {\adj})$ is a tree $(T, {\leq})$ together with a labelling function $\lbfn : T \rarr \pow(V)$ satisfying the following conditions: 
\begin{itemize}
\item (TD1) for all $v \in V$, for some $x \in T$, $v \in \lbfn(x)$; 
\item (TD2) if $v \adj v'$, then for some $x \in T$, $\{ v, v' \} \subseteq \lbfn(x)$; 
\item (TD3) if $v \in \lbfn(x) \cap \lbfn(x')$, then for all $y \in \pth(x, x')$, $v \in \lbfn(y)$. 
\end{itemize}
The width of a tree decomposition is given by $\max_{x \in T} |\lbfn(x)| -1$. 

We define the tree-width  $\tw(G)$ of a graph $G$ \cite{robertson1986graph} as $\min_{T} \mathsf{width}(T)$, where $T$ ranges over tree decompositions of $G$.
This parameter plays a fundamental role in combinatorics, algorithms and parameterized complexity.

\subsection*{Tree-width and pebbling}

We shall now give an alternative formulation of tree-width which will provide a useful bridge to the coalgebraic characterization. 
It is also interesting in its own right: it clarifies the relationship between tree-width and tree-depth, and shows how pebbling arises naturally in connection with tree-width.

A $k$-pebble forest cover for a graph $G = (V, {\adj})$ is a forest cover $(V, {\leq})$ together with a pebbling function $p : V \to \kset$ such that, if $v \adj v'$ with $v \leq v'$, then for all $w \in (v,v']$, $p(v) \neq p(w)$.

\begin{theorem}
Let $G$ be a finite graph. The following are equivalent:
\begin{enumerate}
\item $G$ has a tree decomposition of width $< k$.
\item $G$ has a $k$-pebble forest cover.
\end{enumerate}
\end{theorem}
For the proof of this result, see \cite{abramsky2017pebbling}.

\subsection*{Treewidth as coalgebra number}

The following correspondence between coalgebras and $k$-pebble forest covers is proved similarly to Theorem~\ref{fcth}.

\begin{theorem}
Let $\As$ be a finite $\sg$-structure. There is a bijective correspondence between:
\begin{enumerate}
\item $\Pk$-coalgebras $\alpha : \As \to \Pk \As$
\item $k$-pebble forest covers of $\Gf(\As)$.
\end{enumerate}
\end{theorem}

\noindent We write $\cnP(\As)$ for the coalgebra number of $\As$ with respect to the the pebbling comonad.
The following is a straightforward consequence of the previous two results.

\begin{theorem}
For all finite structures $\As$: $\tw(\As) \, = \, \cnP(\As) - 1$.
\end{theorem}

\section{Further Developments: Rossman's Theorem}

We shall give a preliminary account of ongoing joint work with Tom Paine which aims to apply structural methods to Rossman's Homomorphism Preservation Theorems \cite{Rossman2008}, which are major results in finite model theory.
The proofs in \cite{Rossman2008} are intricate, and use elaborate towers of  ad hoc definitions. The aim of our work is to give a more structural account of these results, using category theory. By doing so, we make clear how standard constructions are being used, in a manner amenable to generalization.

We shall only discuss the following theorem from \cite{Rossman2008}.
\begin{theorem}[Equirank Homomorphism Preservation Theorem]
A formula is preserved by homomorphisms iff it is equivalent, over all structures, to an existential positive formula of the same quantifier rank.
\end{theorem}
Note that this is a weaker result than the finite HPT, since the equivalence is over all structures, not just finite ones.

Rossman makes extensive use of relations $\arX$, $\arnX$, $\barX$, $\barnX$. These are relations between structures $A$, $B$ over some relational vocabulary $\sigma$:
\begin{itemize}
\item $A \arX B$ holds, for a set $X \subseteq A \cap B$, if there is a homomorphism $h : A \to B$ which fixes $X$.
\item $A \arnX B$ holds, if for all finite $C$ of tree-depth $\leq n$ \cite{nevsetvril2006tree}, $C \arX A$ implies $C \arX B$.
\item $A \barX B$ if $A \arX B$ and $B \arX A$; and similarly for $A \barnX B$.
\end{itemize}
We omit the subscript in the case $X = \emp$.

The key result towards the Equirank Homomorphism Preservation Theorem in  \cite{Rossman2008} is the following
\begin{theorem}[Combinatorial Equirank HPT]
For all structures $A$, $B$ with $A \rln B$, there are structures $A'$, $B'$, such that $A \rla A' \equiv^n B' \rla B$.
Here $\equiv^n$ is elementary equivalence with respect to first-order formulas of quantifier rank $\leq n$.
\end{theorem}

The Equirank HPT theorem is a fairly straightforward consequence of this result.

\subsection*{Categorical ingredients}
We find that the main elements of Rossman's argument are built from standard categorical ingredients: coslice categories, comonads, and colimits.
The advantage of this analysis is that we get a more robust version of the construction, which we can adapt e.g.~to the pebbling comonad.

A first  point is that we replace the arrows $\arX$ by arrows in coslice categories:
\begin{center}
\begin{tikzcd}[column sep=small]
& K \ar[dl, rightarrowtail, "e_{1}"'] \ar[dr, rightarrowtail, "e_{2}"] &  \\
A    \ar[rr, "h"'] & & B 
\end{tikzcd} 
\end{center}
In particular, we are interested in objects under injective morphisms $K \emb A$, with $K$ a finite structure in which all the relations have the empty interpretation.

We can then lift the EF comonad to the coslice categories.

\[ \begin{tikzcd}
& K \ar[dl, rightarrowtail, "e_{1}"'] \ar[dr, rightarrowtail, "e_{2}"] &  \\
A    \ar[rr,  "f", "n" very near end] \ar[rr, "g"', leftarrow, shift left=-1ex, "n"' very near start] & & B 
\end{tikzcd} 
\]

This gives us coKleisli maps both ways extending a given partial isomorphism between $A$ and $B$.

\subsection*{Extendability}

A structure $A$ is \emph{$n$-extendable} if for all structures $B$, and $K \emb K' \in \CSn$:
\[ A \barnK B \IMP A \barn_{K'} B . \]
More explicitly, this says
\[ \begin{tikzcd}
& K \ar[dl, rightarrowtail, "e_{1}"'] \ar[dr, rightarrowtail, "e_{2}"] &  \\
A    \ar[rr,  "f", "n" very near end] \ar[rr, "g"', leftarrow, shift left=-1ex, "n"' very near start] & & B 
\end{tikzcd} 
\quad \IMP \quad
\begin{tikzcd}[column sep=large]
& K' \ar[dl, dashrightarrow, "e_{3}"'] \ar[dr, rightarrowtail, "e_{4}"] \ar[r, leftarrowtail, "e"] &  K \ar[d, rightarrowtail, "e_2"] \\
A    \ar[rr,  dashrightarrow, "f'", "n" very near end] \ar[rr, "g'"', dashleftarrow, shift left=-1ex, "n"' very near start] & & B 
\end{tikzcd} 
\]
The dashed arrows are those which are claimed to exist, given the others.

The importance of this notion is that it provides a bridge from the coarse equivalence $\barn$ to the much finer equivalence $\equiv^n$.
\begin{proposition}
If $A$ and $B$ are $n$-extendable, then $A \barn B \IMP A \equiv^n B$.
\end{proposition}

\subsection*{The Main Construction}

Given a structure $A$, we shall define a structure $\EC(A)$, as the colimit of the following ``wide pushout'' diagram $D_A$. We have a single copy of $A$, and for every coslice morphism 
\begin{center}
\begin{tikzcd}[column sep=small]
& K \ar[dl, rightarrowtail, "e_{1}"'] \ar[dr, rightarrowtail, "e_{2}"] &  \\
C    \ar[rr, "h"'] & & A 
\end{tikzcd} 
\end{center}
with $(C, h) \in \Approxn(A)$, the arrow $\iKCh : A \emb C \poK A$. Here $\Approxn(A)$ is a ``suitable class'' of homomorphisms with codomain $A$.

While $D_A$ is, \textit{prima facie}, a large diagram, by taking representatives of isomorphism equivalence classes we can make it  small. With more effort, we could even make it finite; however, our final construction will in any case be infinite.

\begin{proposition}
$A$ is a retract of $\EC(A)$. That is, there are homomorphisms $i_A : A \to \EC(A)$ and $r_A : \EC(A) \to A$ such that $r_A \circ i_A = \id_{A}$.

Moreover, $\EC$ extends to a functor, and these retractions and coretractions are natural.
\end{proposition}

\subsection*{Final step}

We are not there yet! 
$\EC(A)$ extends $A$, but is not extendable in itself.
To get an extendable structure with $A$ as a retract, we have to take the \emph{initial fixpoint} of $\EC$ above $A$, given by the colimit of
\[ \begin{tikzcd}
A \ar[r, "i_A"] & \EC(A) \ar[r, "i_{\EC(A)}"] & \EC^2(A) \ar[r, "i_{\EC^2(A)}"] & \cdots
\end{tikzcd}
\]
Doing this yields the Equirank HPT.
Moreover, we can adapt all this to the pebbling comonad, yielding an ``Equivariable HPT'' (in bounded quantifier rank).
The main point is to use a modified notion of extendability. Details of this work will appear in a forthcoming paper.

Further questions:

\begin{center}
\fbox{Can we refine this construction to get \emph{finite} equirank and equivariable HPT's?}
\end{center}

\section{Further Horizons}

While we are not aware of any closely related work, the work of Bojanczyk \cite{bojanczyk2015recognisable} and of Adamek \textit{et al} \cite{urbat2016one} on recognizable languages over monads should be  mentioned.
The aim of these works is to use monads as a unifying notion for the many variations on the theme of recognizability.
Another line of work in a kindred spirit is by Gehrke and her coauthors on a topological and duality-theoretic approach to generalized recognizability, with applications to descriptive circuit complexity \cite{DBLP:conf/icalp/GehrkePR16,DBLP:conf/lics/GehrkePR17}.

There are many opportunities for further development of our approach, which will contribute to relating \Struct~and \Power.
In particular, major algorithmic metatheorems such as Courcelle's theorem \cite{courcelle1990monadic}, and 
decision procedures for guarded fragments \cite{gradel1999decision}, offer promising avenues for progress.

\begin{itemize}
\item Can we combine coalgebras witnessing treewidth/tree model property with coalgebraic descriptions of automata, to obtain structural -- and generalizable -- proofs of these results?

\item Can we embed our coalgebraic semantics in a type theory, allowing the algorithms whose existence is asserted by these results to be extracted by Curry-Howard?

\item The wider issue: can we get \Struct~and \Power~to work with each other to address genuinely deep questions?
\end{itemize}

\subsection*{Envoi}
We hope that the material surveyed in this short note is appropriate in spirit for a volume dedicated to Maurice Nivat.
\begin{itemize}
\item We note firstly that Nivat is an ancestral pioneer both for \Power, in his work on formal languages and automata, and for \Struct, in his work on semantics.
\item Much of his work in both fields concerned words and trees, finite and infinite, and various notions of approximation \cite{nivat1979infinite,arnold1980metric}. Our comonads are, concretely, built out of words (or sequences) and operations on them; while more globally, they have a tree structure. For example, the counit map $\Ek \As \to \As$ ``covers'' the structure $\As$ with a tree-structured model; and the levels $k$ indexing the comonads give suitable notions of approximation.
\item We hope that our work is in the spirit of his concern with the fundamental mathematical structures of computation.
\end{itemize}
We hope, finally,  that he would have approved of the message with which we would like conclude:
\begin{center}
\textbf{\Large Let's not forget to dream!}
\end{center}

\bibliographystyle{alpha}
\bibliography{bibfile}

\end{document}